%% file: main.tex
\documentclass[10pt,conference]{ieeeconf}

\usepackage{header}
\usepackage{epstopdf}
\title{How Much Traffic is Searching for Parking? Simulating Curbside Parking as a Network of Finite Capacity Queues}

\author{Chase Dowling, Tanner Fiez, Lillian Ratliff, and Baosen Zhang
\thanks{The authors have been supported in part by NSF grants CNS-1646912 and CNS-1634136. C. Dowling was also supported in part by the Washington Clean Energy Institute.}
\thanks{C. Dowling, T. Fiez, L. Ratliff and B. Zhang are with the Department of Electrical Engineering, University of Washington,
	    Seattle, WA 98195, USA
	    Emails: \{cdowling,fiezt,ratliffl,zhangbao\}@uw.edu}}

\IEEEoverridecommandlockouts

\begin{document}

\maketitle
\thispagestyle{empty}
\pagestyle{empty}

\begin{abstract}
\input{abstract}
\end{abstract}

\section{Introduction}
\input{intro}

\section{Background}
\label{sec:background}
\input{background}

\section{Results}
\label{sec:results}
\input{results}

\section{Application}
\label{sec:application}
\input{application}

\section{Conclusion}
\label{sec:conc}
\input{conclusion}

\bibliographystyle{IEEEtran}
\bibliography{main} 

\end{document}

%% file: abstract.tex
With the increasing availability of transaction data collected by digital parking meters, paid curbside parking can be advantageously modeled as a network of interdependent queues. In this article we introduce methods for analyzing a special class of networks of finite capacity queues,  where tasks arrive from an exogenous source, join the queue if there is an available server or are rejected and move to another queue in search of service according to the network topology. Such networks can be useful for modeling curbside parking since queues in the network perform the same function and drivers searching for an available server are under combinatorial constraints and jockeying is not instantaneous. Further, we provide a motivating example for such networks of finite capacity queues in the context of drivers searching for parking in the neighborhood of Belltown in Seattle, Washington, USA. Lastly, since the stationary distribution of such networks used to model parking are difficult to satisfactorily characterize, we also introduce a simulation tool for the purpose of testing the assumptions made to estimate interesting performance metrics. Our results suggest that a Markovian relaxation of the problem when solving for the mean rate metrics is comparable to deterministic service times reflective of a driver's tendency to park for the maximum allowable time.

%% file: intro.tex
Since the advent of digital parking meters, cities have stockpiled a growing record of parking transaction data within their CBD. Transaction data provides engineers with a means of estimating the arrival rate of drivers which \emph{attains} the observed occupancy level. We can combine this rate information with a queue-theoretic model of downtown parking, necessitating an evaluation of sufficient conditions as well as assumptions made when solving for the stationary distribution. Consider a block-face of curbside parking spaces (see Fig.\ref{fig:blockface}): this represents a finite capacity queue with no buffer. Drivers that arrive in search of parking that find all spaces occupied must move onto an adjacent block-face. This search dynamic driven by the rate of drivers turned away from a full block-face is representative of the impact of drivers searching for parking, and hence curb-side parking impact on through-traffic.

This model has been recently introduced in \cite{dowling2017optimizing}, but a number of assumptions are made. The primary contributions of this work are 1) sufficient conditions for the number users in such a system not growing unboundedly, 2) simulated analysis of rates of convergence to an apparent steady-state, and 3) simulating a real-world network of curbside parking to evaluate these assumptions.

The rest of this paper is organized as follows: we provide background, notational preliminaries, and information about our simulator architecture in Sec.~\ref{sec:background}, state our results on sufficient conditions (Sec.~\ref{sec:results}), and steady-state convergence rates (Sec.~\ref{sec:convgrate}). In Sec.~\ref{sec:application} we set up our use case example: providing information on data sources (Sec.~\ref{sec:data}) and simulation results (Sec.~\ref{sec:simapp}). We make our concluding remarks in Sec.~\ref{sec:conc}.

%% file: background.tex
\begin{figure}
    \centering
    \includegraphics[width=0.9\columnwidth]{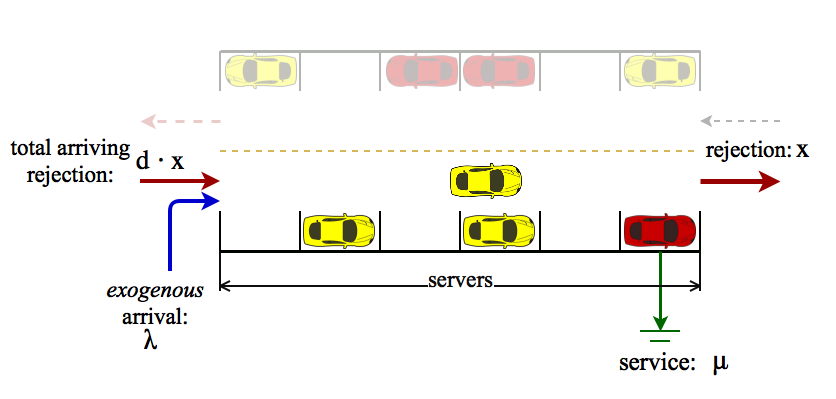}
    \caption{A block-face of curbside parking represented as a finite capacity queue.}
    \label{fig:blockface}
\end{figure}

Canonical queueing networks like Jackson \cite{jackson2004jobshop} or (in general) BCMP networks \cite{baskett1975open} operate under a regime where tasks join the network at some queue, are served, and then move onto the next queue according to the network topology or exit according to some probability. Characterized by mild conditions on the distributions of their arrival and service rates, the state spaces of BCMP networks are seperable, greatly improving the tractability of their analysis.

We consider a new service regime where a network of queues has some exogenous arrival process and if a task arrives at a queue without an available server, the task searches according to a network topology, but requirements for separability via the BCMP theorem are not met (e.g., general instead of negative exponential service time distributions) due to physical drivers in a system. Once the task is served at a single queue somewhere in the network, it immediately exits the network after service.

Parking in the core business districts (CBD) of cities is naturally amenable to analysis via such networks of finite capacity queues with the growing availability of parking transaction data: parking queue tasks are vehicles in need of a space to park and servers are those spaces. Service time distributions can be characterized by the length of paid parking time. And the finite capacity of curbside and garage or lot parking is all too apparent to the individual driver in search of a space when supply is scarce. Further, in transportation literature, \emph{cruising for parking} also results from drivers searching for an available curbside parking space to avoid garage prices \cite{arnott2006spatial}. In the aggregate this behavior creates potentially significant congestion \cite{shoup:2005aa}, but city planners have until recently lacked high resolution (block-face by block-face, per hour) models of such costs \cite{ratliff:2016aa}. Congestion caused by drivers cruising for parking is non-trivial, historically cited as composing up to 30\% of through-traffic \cite{inci:2015aa,shoup:2006aa}. 


Supposing vehicles enter a network of parking queues in an effort to find a space, park, and then exit the network after some amount of time, queueing networks with seperable state-spaces may not be suited to describing the state space of these parking queues; this is supported by evidence of the probabilistic dependence of adjacent block-faces of curbside parking \cite{fiez:2017aa} (i.e. a block-face of curbside parking that is full is unlikely to be adjacent to an empty block-face), in addition to a number of other factors we will describe. We are able to address this, however, through the use of some relaxing assumptions and we verify these assumptions via simulation. The purpose of this paper is to investigate the validity of these assumptions.

\subsection{Related Work: Queueing Theory}

BCMP networks with finite queueing capacities have been analyzed by incorporating some blocking probability at each queue, or by allowing tasks to be dropped once the capacity of the queue is reached \cite{balsamo2003review,balsamo2011queueing}. In some contexts like parking, this may be an unreasonable assumption. Consider the case of parking spaces in a CBD: drivers can neither a) be held in place by some blocking protocol or b) disappear from the network while in search of a parking space either curbside or garage. A vehicle, or queue task, constantly impacts the performance of the system.

Intuitively, the queuing regime we are interested in is more akin to jockeying than current research in networks of finite capacity queues. In typical jockeying problems, tasks switch between queues or servers based on a jockeying strategy (e.g. probabilistic or rule-based strategies) with the motivation, in practice, being a shorter sojourn time \cite{koenigsberg1966jockeying}. In our motivating case, drivers are forced to search between queues until an available server is found but in a combinatorially constrained fashion---drivers may only search a limited set of block-faces with each trial based on the connectivity of the network.

\subsection{Related Work: Parking}

Canonical models for parking tend to assume a degree of homogeneity \cite{douglas:1975aa, arnott2006spatial, arnott:2009aa}, but this limitation was largely a function of the availability of data on curbside parking occupancy: namely the proportion of spaces in use at any given time. This data has traditionally be collected by manually \cite{inci:2015aa, shoup:2007aa}, but the introduction of digital parking meters has provided researchers with an opportunity to increase the spatial and temporal resolution of CBD parking models. Indeed, a growing body of work is beginning to make use of parking transaction data as a means to estimate curbside parking occupancy \cite{yang2017turning,fiez:2017aa} in addition to investigating price elasticity of demand \cite{pierce2013getting}. In \cite{dowling2017optimizing}, a queueing network model designed specifically to take advantage of these new sources of data is initially introduced, relating occupancy to price, but a number of assumptions are made.

Queues are not new to traffic engineers: queues have been used to analyze the flow of traffic along a roadway \cite{papageorgiou2003review} or through a signalized intersection \cite{newell1965approximation}. In an attempt to capture the parking-congestion relationship, several approaches based on queuing theory have been previously introduced~\cite{bender1997simulating,klappenecker2014finding,arnott:1999aa,portilla:2009aa,larson:2010aa,ratliff:2016aa} where roads (or segments), parking spaces, or both are modeled as queues.

Previous applications of queuing theory to curbside parking have been focused on investigating the short-term impact on through-traffic or an intersection due to drivers maneuvering into a parking space \cite{portilla:2009aa, cao2015generalized}. To be clear, this work is interested in longer, steady-state analysis of curbside parking resource performance and its impact on expected traffic volumes; long-term performance metrics like occupancy drive policy decisions like price and maximum parking time. Steady-state analysis of garage or lot parking modeled as queues has also recently begun to appear, but these are treated as a single queue with many servers \cite{caliskan2007predicting}, and congestion resulting from finite supply is not considered.

\subsection{Preliminaries}
A queue, or a vertex $i \in V$ is characterized by an \emph{exogenous} arrival rate $\lambda_i$, a service rate $\mu_i$, the number of servers $k_i$, and maximum number of tasks in the queue $n_i$. We assume that the \emph{exogenous} arrival process is Poisson~(independent between queues) and the services times are generally IID like conventional $M/GI/\cdot/\cdot$ queues~\cite{wolff1989stochastic}, however, unlike conventional queueing networks or even traditional finite capacity queue networks where tasks are buffered or blocked at or by individual queues, the \emph{network edges} themselves form queues with some pre-determined travel (service) time. For simplicity, we assume network edge queues are infinite server, first-come first-served (FCFS) queues with fixed travel time since most of our application focus area of Belltown is made up of uncongested side streets. More realistic queue models reflective of the traffic state could be used. At each node that is reached by a task, tasks assess whether a server is available or continue to search.

The key difference between the proposed queue network and conventional networks--such as a Jackson network~\cite{jackson1957networks}--is that tasks proceed to other queues after they are rejected rather than served. Since the rejection of a queue with Poisson arrivals and exponential service times is not Poisson, characterizing the stationary distribution of this network of queues is difficult because the distribution of total arrival rate itself to any vertex queue is unknown. Further, the service rates at these queues are generally distributed, failing to meet criteria for separability according to the BCMP theorem. As will be clear in our application, negative exponentially distributed service times at each finite capacity queue is likely too strong an assumption (see Fig.~\ref{fig:servicedist} for a distribution of paid parking transaction times).

In Kendall's notation, $M/GI/k/n$ queues have Poisson arrival, generally distributed but independent service times (M for exponential, D for deterministic), $k$ servers and $n - k$, $n \geq k$, spaces available for tasks to queue. First we'll examine networks of such queues in a uniform, symmetric network. We'll then extrapolate this result to trees. Lastly we'll consider the analysis of general networks, focusing on the $M/M/k/k$ and $M/D/k/k$ service regimes.

\subsubsection{Stationary Distribution of a Single $M/M/k/k$ Queue}
Here we introduce how a \emph{single queue} with exponential service times with finite servers can be analyzed given occupancy data, calling on this later in the paper extending the analysis to a network of queues. To help avoid confusion between exogenous arrivals~(from outside of the network, denoted by $\lambda$) and endogenous arrivals~(rate of rejection from neighboring queues, denoted by $x$), we use  $y = \lambda + x$ as the total arrival rate to a queue. Suppose the service rate~(inverse length of parking time) of each server is $\frac{1}{\mu}$ and there are $k$ servers~($k$ parking spots) in total. Let $\pi_i$ be the stationary probability that $i$ servers are busy~($i$ cars are parked), for $i=0,\dots,k$. Let $\boldsymbol{\pi}=[\pi_0 \; \dots \pi_k]$. For this single queue, we can explicitly write down its stationary probability distribution via the transition rate matrix:

\begin{equation*}\label{eqn:transition}
   \bd Q= \bma -y & y & 0 & \cdots & 0\\
   \mu & -(\mu+y) & y & \cdots & 0\\
   \vdots & & \ddots & & \vdots \\
   0 & \cdots & (k-1)\mu & -((k-1)\mu + y) & y \\
   0 & 0 & \cdots & k\mu & -k\mu \ebma,
\end{equation*}
and $\boldsymbol{\pi}$ is the unique solution to
\begin{equation}\label{eqn:stationary}
\boldsymbol{\pi Q} = \boldsymbol{0}
\end{equation}
such that $\sum \pi_i =1$. Let $\rho = \frac{y}{\mu}$. By standard calculations~\cite{wolff1989stochastic},
\begin{equation}\label{eqn:statvalue}
\boldsymbol{\pi} = \pi_0 \cdot \left[ 1, \rho, \cdots, \frac{\rho^{k}}{k!} \right]
\end{equation}
where $\pi_0 = [\sum_{j = 0}^{k} \frac{\rho^{j}}{j!}]^{-1}$. Using Little's Law, the occupancy $u$, or the proportion of busy servers at any given time can be expressed as,
\begin{equation}\label{eqn:little}
u = \frac{y}{k\mu}\left( 1 - \pi_0 \frac{\rho^{k}}{k!} \right)
\end{equation}.

Little's Law does not depend on the distributions of the service time and arrival process; as we've used it in Eqn.~\eqref{eqn:little}, we merely need to be able to estimate the probability the block-face is full.

Note that $(1 - \pi_0 \frac{y^{k}}{k!})$ is the probability that \emph{at least} one space is available. Consider, if drivers are unable to wait for an available server at a particular block, in order to obtain occupancies approaching 100\%, cars would need to arrive at an infinite rate in order to immediately replace vehicles exiting service.

A block-face queue is therefore rejecting incoming vehicles at a rate of $y\cdot\pi_k$. The difficulty therein lies with estimating these total arrival rates, because no two adjacent block-faces are independent as long as travel time between queues is finite.

\subsection{Simulator Architecture}
\label{sec:simulator}

Our simulator is written in Python and is freely available to download and test at \url{github.com/cpatdowling/net-queue}\footnote{Our experiments utilize GNU Parallel \cite{Tange2011a}.}. Requirements and basic instructions, as well as data and relevant parameters used in this paper are included in the repository. The simulator constructs a network of block-face (drivers in service/parked) and street (drivers moving from one block-face to the next) queues linked according to the true street topology. Our simulator is validated against occupancy data provided by the Seattle Department of Transportation (SDOT).

In line with our model design the simulator treats streets and block-faces individually: once a driver reaches the end of their drive time on a street, they ``immediately check'' the entire block-face they've arrived at for availability. If no parking is available, the driver chooses a new destination uniformly at random (though more realistic search strategies can be applied \cite{hampshire:2016aa}) based on the block-faces currently accessible to them according to the street topology, joining a street queue with some driving service time associated with it.

The input parameters of our simulator include:
\begin{enumerate}
\item \textbf{Network Topology.} For every block-face in Belltown, there will be any number of block-faces a driver can reach using only legal maneuvers on one- and two-way streets, excluding legal U-turns\footnote{Our data and roadway maps currently provide no principled. The simulator is given a map of block-face connectivity transcribed from Google Maps}.

\item \textbf{Service Rate}: The inter-service time dictates how long cars will spend parked on a block-face. Fig.~\ref{fig:servicedist} illustrates the distribution of \emph{paid} parking times across Belltown, between block-faces set for 2 and 4 hour maximum parking times. In simulation, this distribution can be set as exponential, deterministic, or uniquely determined by the distribution of paid transaction times exhibited at the block-face level. At this point we have no reliable data on the frequency of illegally parked vehicles that have either paid, or overstayed, and further we have no means of measuring how early drivers typically leave before their paid time expires. In our initial simulations, we assume everyone parks legally and early/late departures balance out; the latter assumption is not unfounded and studied at length in \cite{yang2017turning}.
\item \textbf{Number of Servers}: The number of parking spaces, or the number of servers in the block-face queue, are extracted directly from data for each block-face, ranging in values according to Fig.~\ref{fig:supplydist}. In our data, it is sometimes the case that there may be higher than 100\% occupancy at any given block-face as a result of factors described in Sec.~\ref{sec:data}. In these cases we assume that occupancy is 100\% with respect to the estimated number of \emph{spaces}, and not with respect to the number of vehicles currently in service.
\item \textbf{Exogenous Arrival Rate}: The simulator accepts a mean parameter for an exponential random inter-arrival time distribution, simulating vehicles arriving at a specific block-face to begin their search for parking. If a space is available at the block-face they originally arrive at, the driver accepts the first space without contributing to congestion. 
\item \textbf{Drive Time}: Drivers arrive at a block-face and determine if any spaces are available. If no spaces are available, a drive time is specified to determine how long it takes drivers to reach the next adjacent block-face in their search.
\end{enumerate}

Important output values of our simulator include:
\begin{enumerate}[noitemsep,nolistsep]
\item \textbf{Traffic due to Parking (rate of rejections)}: Traffic due to drivers searching for parking can be measured as the total number of rejections at a particular block-face (or road, if search strategy is non-uniform) or as rejections per unit time.
\item \textbf{Average Wait}: The amount of time a driver spends looking for parking, as a function of drive time between each block-face a driver is
    rejected from until they find parking.
\item \textbf{Occupancy}: The resulting occupancy measures the average number of servers or spaces along a block-face in use at any given time. This value is compared against true occupancy data to ensure the simulator is providing accurate estimates of congestion and sojourn times.
\end{enumerate}

%% file: results.tex
Here we provide some analytical results on networks of finite capacity queues.

Each vertex in the network is a finite-capacity, multi-server queue. The queues are connected as \emph{vertices on a directed graph}. We use conventional notations $G=(V,E)$ to describe this digraph. Assumptions regarding connectivity are made accordingly in our results. Each vertex $i$ has some exogenous arrival rate $\lambda_i$ and service time $\mu_i$. When a task is rejected from a queue at vertex $i$, it transits along edge $(i,j)$, where $j$ is connected to $i$, denoted $j \sim i$. . 

For simplicity, we consider this edge to be an infinite server queue with a fixed travel time $d$ shared across all edges. As we will see, if a travel time is not imposed along tasks transiting between $(i,j)$ then the combinatorial search constraints modeled by the connectivity of the graph is ill-defined.

To gain some intuition into how this system behaves, we first state a lemma.

\begin{definition}
A network of finite capacity queues ``communicate'' if a queue at vertex $i$ is reachable by a task from queue $j, \forall i, j \in V$.
\end{definition}

\begin{lemma}\label{lemma:capacity}
Given travel time $d > 0$, and all queues in the network $G$ communicate, if 

\begin{equation}\label{eqn:capcond}
\sum_{i} \lambda_i < \sum_{i} \mu_i,
\end{equation}

for $M_i/GI_i/k_i/k_i$ queues, then the number of tasks in the system does not grow unboundedly.
\end{lemma}

\begin{proof}
Since each vertex queue has finite capacity, we only need to show that the number of tasks in the infinite capacity edge queues does not grow unboundedly.

Let $d = \delta$ be arbitrarily small. Since the network communicates, the cover time (the expected number of steps before a task reaches every vertex in the graph) is upper bounded by $O(n^3)$ \cite{lovasz1993random} in the number of vertices $n$. The total, worst-case travel time to traverse the network $\epsilon = \delta O(n^3)$.

For small $\epsilon$, the queue network $G$ becomes a bulk infinite capacity queue $\boldsymbol{M}/\boldsymbol{GI}/\sum_{i} k_i/ \infty$, where the total arrival rate $\Lambda$ is the sum of Poisson processes parameterized by rates $\lambda_i$. Since, 

\begin{equation}
\Lambda < \sum_{i} \mu_i
\end{equation}

then by standard conditions \cite{kingman2009first}, this bulk queue $\boldsymbol{M}/\boldsymbol{GI}/\sum_{i} k_i/ \infty$ is stable, and therefore the number of tasks in the system do not grow unboundedly.
\end{proof}


In other words, service capacity is greater than the total exogenous arrival rate. We can further intuit the implications of this lemma. Consider a two-node finite capacity queue network with arbitrarily short travel times. During periods in which both queues are busy, new exogenous arrivals are rejected and traverse back and forth between the two queues asymptotically quickly at a rate proportional to $1/2d$. Lemma \ref{lemma:capacity} does not imply that the \emph{rate} of rejection---a primary quantity of interest---from a queue in the network is bounded. The observation that rejection rates grow asymptotically is valuable in the context of our application in Sec.~\ref{sec:application}.

\subsection{Two-node network}

We can attempt to solve for the stationary distribution of a completely connected, two-node network with a single server at each queue (see Fig.~\ref{fig:twonodenet}). This infinite state space for the single server case is represented in Fig.~\ref{fig:statespace}. Since the exact stationary distribution of the network is difficult to characterize, we instead turn to understanding the behavior of the \emph{mean performance metrics} of the network.

\begin{figure}
    \centering
    \includegraphics[width=0.9\columnwidth]{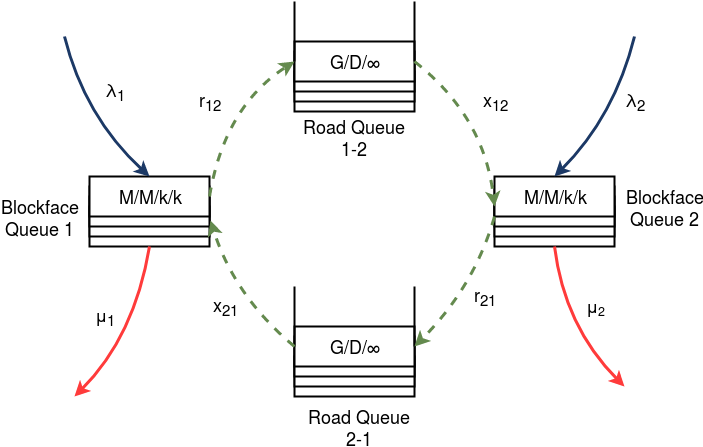}
    \caption{2 node network}
    \label{fig:twonodenet}
\end{figure}

\begin{figure}
    \centering
    \includegraphics[width=0.9\columnwidth]{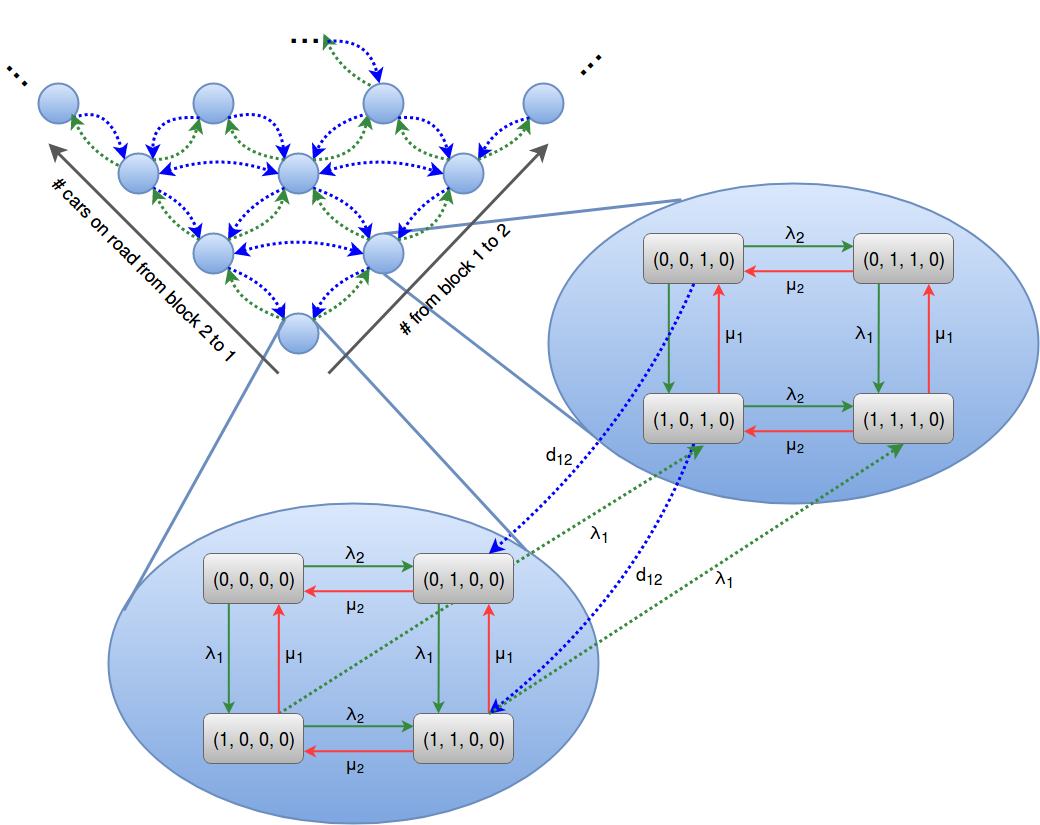}
    \caption{2 node network state space}
    \label{fig:statespace}
\end{figure}

If we define the rate as the number of vehicles that passes through a point over a unit period of time, then at steady state, the input rate and output rate of a road has to be the same. The output never more than the input since edges in the network do not have exogenous arrivals, and if the output is less than the input vehicles are queued up in the road. This allows us to restrict our vision the state-space described by the root node in Fig.~\ref{fig:statespace} and relax the problem to a Markovian setting. This network can then be represented by the rate transition matrix $Q$ as, 

\footnotesize
\begin{equation}
  \bma -\sum_{j \neq 1} Q_{1,j} & \la_1+x_{12} & \la_2+x_{21} & 0 \\
  \mu_1 & -\sum_{j \neq 2} Q_{2,j} & 0 & \la_2+x_{21} \\
  \mu_2 & 0 & -\sum_{j \neq 3} Q_{3,j} & \la_1+x_{12} \\
  0 & \mu_2 & \mu_1 & -\sum_{j \neq 4} Q_{4,j}
  \ebma.
\end{equation}
\normalsize

The stationary probabilities $\mathbb{\pi}$ are the solution to,

\begin{equation}
\mathbb{\pi}Q = 0.
\end{equation}

Define $P_1 = \pi_1 + \pi_2$ and $P_2 = \pi_2 + \pi_3$ to be the probabilities that queues 1 and 2 in the network are full. Suppose both roads have the same travel time $d$, and let $W_1$ and $W_2$ be the sojourn time of a vehicle that enters node 1 and node 2 first, respectively. We have

\begin{subequations}
  \begin{align}
    W_1 &= P_1^1 (d+ W_2) \\
    W_2 &= P_2^1 (d+ W_1),
  \end{align}
\end{subequations}
and solving we observe that,
\begin{equation}
  W_1= \frac{P_1(1+P_2)}{1-P_1 P_2} \;\; W_2=\frac{P_2(1+P_1)}{1-P_1 P_2}
\end{equation}
and the mean sojourn time assuming equal arrival probability is $1/2(W_1+W_2)$.

In the symmetrical case with $\la_1=\la_2=\la$, $\mu_1=\mu_2=\mu$, we can explicitly calculate $\bm{\pi}$. Normalizing such that $\la=1$, and assuming $\mu>1$ for stability, we have
\begin{equation}
  \bm{\pi} = \begin{bmatrix} \frac{(\mu-1)^2}{\mu^2} & \frac{\mu-1}{\mu^2} & \frac{\mu-1}{\mu^2} & \frac{1}{\mu^2} \end{bmatrix}
\end{equation}
and $x_{12}=x_{21}=\frac{1}{\mu-1}$, with sojourn time $\frac{\mu}{\mu^2-1}d$.

\subsection{Symmetric networks}

Many urban centers have fairly uniform street topologies where the streets from a regular graph. In this section we make the assumption that the queueing network is entirely uniform: the topology is a $d$-regular graph, all block-faces have the same number of servers with the same service rate $\mu$, and they have the same exogenous arrival rate $\lambda$.

In this regular queue network, each queue will have equal stationary distributions in the steady state, therefore we only need to look at a single queue as representative of the state space of the entire network. Let $x$ be the average rate of rejection of a queue to one of its neighbors, and $dx$ be the total rejection to all of its neighbors. Let $y=\lambda+dx$ be the total arrival rate to a queue, where $\lambda$ is the exogenous arrivals and $dx$ are the rejections from its neighboring queues. We have the conservation equation,

\begin{equation}\label{eqn:consv}
    dx = y\pi_k,
\end{equation}
 where $\pi_k$ is the probability that all $k$ severs are busy.

 Combined with stationary distribution of \eqref{eqn:stationary} we have the following equations:
\begin{align}
	\left\{ \begin{array}{ll}
		\boldsymbol{\pi} \boldsymbol{Q} & = 0\\
		\sum \pi_i & = 1\\
		dx&=\pi_k (\la+dx)\end{array}\right.
\end{align}
We can write (\ref{eqn:consv}) as,
\begin{equation} \label{eqn:y}
y-\la= \frac{\frac{\rho^k}{k!}}{\sum_{i=0}^k \frac{\rho^i}{i!}} y
\end{equation}
where $\rho=\frac{y}{\mu}$. The equation in \eqref{eqn:y} is a polynomial in $y$. The next lemma states that there exists a unique solution to $y$ (and thus $x$) as long as the queues are stable:
\begin{lemma}
    \label{thm:lemma1}
 If $0 < \la < \mu k$, then (i) there is a unique and
    positive solution to $y$ in \eqref{eqn:y} and (ii) the solution is greater than $\la$. In addition, the rejection rate
    $x$ is also unique and positive.
    \label{lem:lem1}
\end{lemma}

The result is obtained by observing there is a single sign change in the sequence of coefficients in the polynomial \eqref{eqn:y} and applying Descartes' rule of signs. A complete proof is available in \cite{dowling2017optimizing}. This result states that as long as the total arrivals are less then the service rate times the number of spaces, we can explicitly find the rejection rates and the stationary probabilities by solving a polynomial equation.

\subsection{Irregular networks}

The totally uniform assumption does not pertain to our application, hence the need to test via simulation. But given occupancy data we show that the \emph{total} exogenous and endogenous arrivals to a queue can still be solved for and used to estimate the traffic caused by drivers searching for parking. This time, for some \emph{total} incoming rejection rate $x$, letting $y = \lambda + x$, we can estimate the endogenous proportion of incoming arrivals as the sum of the outgoing fractional rejection rates of adjacent queues. This \emph{decoupling} assumption we make in order to model an irregular network like Belltown is tested via simulation in Sec.~\ref{sec:application}.

Assuming the queueing network reaches steady state,
from the perspective of a single queue in solving \ref{eqn:little} for $\pi_0$ gives
\begin{equation}
\pi_0 \frac{\rho^{k}}{k!}  + \frac{uk\mu}{y} = 1,
\end{equation}
where $u$ is the occupancy level and $\rho=\frac{y}{\mu}$.
Rearranging terms yields a polynomial in $y$,
\begin{equation}\label{eqn:y_poly_occup_mu}
0 = \sum_{i=0}^{k} \frac{1}{\mu^{i-1}}\left[\frac{i - uk}{i!}\right]y^{k}.
\end{equation}
Again, we can characterize the solutions to \eqref{eqn:y_poly_occup_mu}
\begin{lemma}\label{lem:lem2}
If $u \in [0,1)$ and $k$ is a positive integer, then \eqref{eqn:y_poly_occup_mu} has a unique real, positive root.
\end{lemma}

The result is obtained by application of Descartes' rule of signs. A complete proof is available in \cite{dowling2017optimizing}. The analysis of irregular networks is central to our application in Sec.~\ref{sec:application}.

This root need not be bounded, hence the restriction of the values of $u$ to the interval $[0,1)$. In order to achieve a $100\%$ occupancy, implying the probability of being full is $1$, vehicles would need to arrive constantly ($y = \infty$), immediately taking the place of any vehicle that leaves upon completion of service.

\subsection{Simulation of finite capacity queue networks}
\label{sec:convgrate}

Strong results on the stationary distribution of the $M/GI/k/n$ queue are generally not known. As service times tend $\emph{not}$ to be negative exponentially distributed, we can test an the assumption that the network can indeed reach steady state.

Motivated by the observation that parkers tend to park for the maximum allowable time at a location (see Fig.~\ref{fig:servicedist}), we simulate a regular network of 10, completely connected, $M/D/k/k$ queues where service time is deterministic and compare this to an identical network of $M/M/k/k$ queues\footnote{Number of servers per queue: 5, average (or fixed) service times per queue: 5 units time, exponential inter-arrival times: 2.0 units time and 1.2 units time for lower and higher occupancy cases respectively, transit times 0.1 units time}. Fig.~\ref{fig:convrate} illustrates the rates of convergence to an identical level of occupancy at steady-state, where exponential service takes a few time steps longer.

\begin{figure}
    \centering
    \includegraphics[width=0.9\columnwidth]{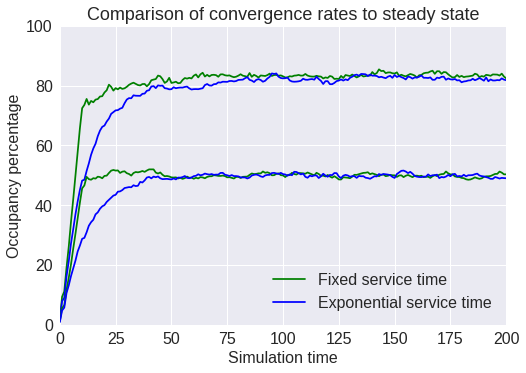}
    \caption{Rate of convergence to steady-state occupancy for a 10-node queue network }
    \label{fig:convrate}
\end{figure}

With 5 servers each with an average service time of 5, we have an effective service rate 1 and the system is stable for exogenous arrival rates below 1 by Lemma \ref{lemma:capacity} (identical across queues for a regular network). For exponential exogenous interarrival times between 4.0 and 1.05 units time, we compare the occupancy, probability the queue is full, and resulting rejection rate between exponential and fixed service times.

\begin{figure}[h]
    \begin{subfigure}[t]{0.45\textwidth}
        \centering
		\includegraphics[width=0.9\textwidth]{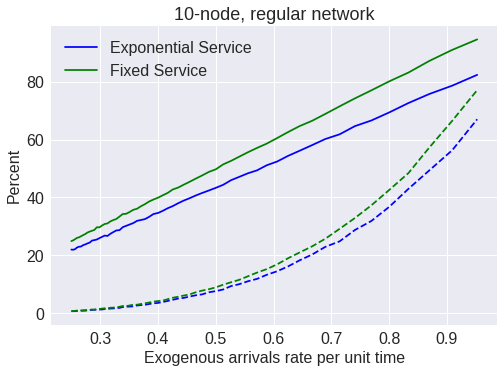}
        \caption{}
        \label{fig:varyoccup}
    \end{subfigure}\hfill
    \begin{subfigure}[t]{0.45\textwidth}
        \centering
        \includegraphics[width=0.9\textwidth]{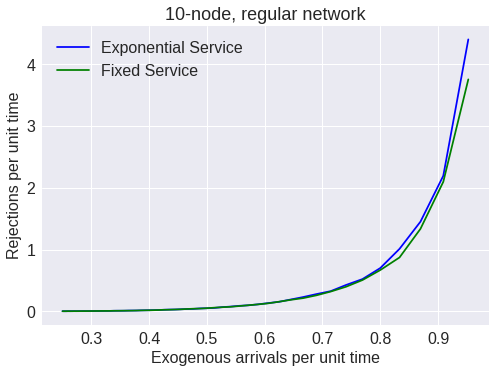}
        \caption{}
        \label{fig:varyrej}
	\end{subfigure}
	\caption{(a) Occupancy (solid) and probability a queue in the network is full (dashed) for increasing arrival rates in a regular 10-node queue network (b) Rejection rate of tasks are unable to find service at queues for increasing arrival rates in a regular 10-node queue network}
	\label{fig:varyarrival}
\end{figure}

Based on Fig.~\ref{fig:varyarrival}, using a Markovian relaxation (assuming negative exponential service time distributions) to solve for metrics of interest in networks of queues with finite capacity (e.g. occupancy, rate of task rejection), provides similar results to a network under deterministic service times. The largest observable difference in regular networks at steady state appears to be the occupancy, as illustrated in Fig.~\ref{fig:varyoccup}, suggesting that probability mass across queue states for an individual queue in the network is skewed towards the block-face being full. Nevertheless, comparing the probability that the block-face is full---the driving state for rejections to occur---are most similar at lower arrival rates, diverging slightly near saturating arrival rates. Rejection rates in simulated regular networks appear to be virtually identical (see Fig.~\ref{fig:varyrej}.

With these empirical results, our model estimates of the probability a block-face is full and the resulting rate of drivers rejected using a Markovian service time solution may not be unfounded for deterministic service rates, though more work will be required to test a mixed service rate distribution closer to the true distribution of paid parking times in Fig.~\ref{fig:servicedist}.

The remaining challenge then is to test the Markovian relaxation achieves similar results to a fixed service time regime in irregular networks. We do so on the network of block-faces representing paid parking in Belltown in the next section.

%% file: application.tex
Searching for parking presents a challenging task in urban districts around the world. Drivers in dense urban areas frequently find that desirable parking close to their destination is unavailable or prohibitively expensive. As a result, the act of \emph{cruising for parking} can arise from any number of situations: desirable parking near a destination being at capacity, price differences between public curbside parking and private garage parking\footnote{The discrepancy between curbside parking and off-street parking can be significant. For example, in some areas of Seattle, parking in a garage costs upwards of \$9/hour compared to the roughly \$2/hour cost of on-street parking~\cite{offstreet:2014aa}. In addition to price discrimination caused by time-dependent fees, an entire day in a garage is approximately \$30. }, or simply a driver's lack of familiarity with their surroundings. Studies have suggested that a majority of drivers spend anywhere between 3.5 to 14 minutes in a typical search~\cite{shoup:2006aa}. These times quickly add up to cause significant productivity losses in cities. For example, a single 15 block district of Los Angeles services over 8,000 cars in day, which leads to 470 to 1870 hours of lost time looking for parking~\cite{shoup:2007aa}. 

\begin{figure}[h]
    \centering
    \includegraphics[width=0.9\columnwidth]{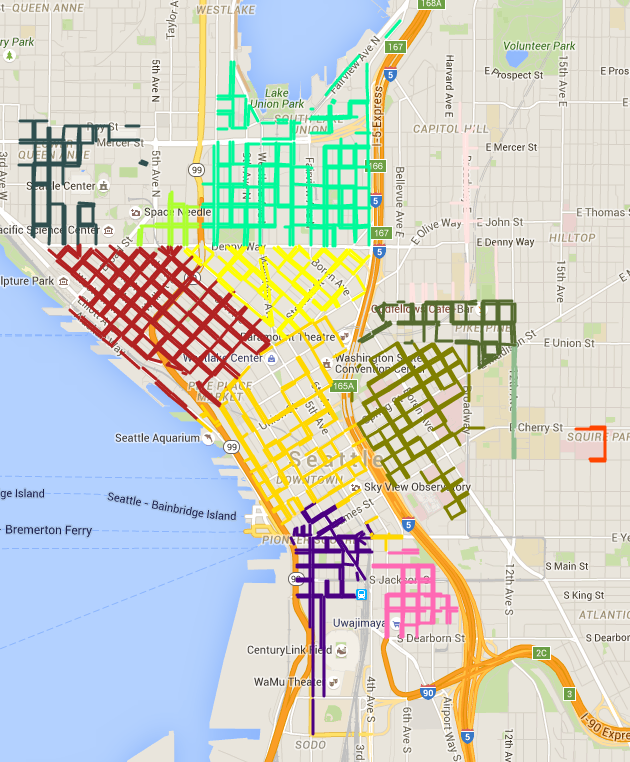}
    \caption{A map of all metered curbside parking in the CBD of Seattle, Washington, USA. Belltown is in dark red along the waterfront. Map data \copyright 2018 Google}
    \label{fig:map}
\end{figure}

We use paid parking transaction data provided by SDOT, collected at digitial parking meters. These types of data are becoming widely available in many cities around the United States. Recent initiatives---LA Express Park in Los Angeles~\cite{ladot:2017aa} and SFpark in San Francisco~\cite{sfmta:2017aa}, for example---are providing both city planners and researchers with a wealth of new data. SFpark is a now concluded pilot study that evaluated the effectiveness of spatially and temporally adjusted pricing for on-street and off-street parking\footnote{The pilot study was conducted on approximately 25\% of SF's smart meters and due to its success, the program will be rolled out across SF.}.  Similarly, LA Express Park is an ongoing program that utilizes smart technologies and demand-based pricing to manage parking in downtown LA.

\subsection{Data}
\label{sec:data}

We utilize on-street paid parking transaction data collected from March 1$^{st}$, 2016 through July 31$^{st}$, 2016 by the SDOT to inform our model. The paid parking transaction data includes both pay-station and pay-by-phone records at a block-face level spatial granularity. In Belltown, there is a total of 256 block-faces across the neighborhood each with a number spaces range of one to 20 parking spaces (distributed as Fig.~\ref{fig:supplydist}). Spaces are not demarcated, as parking is paid for at a digital meter and a permit is displayed in the vehicle's passenger window. To estimate supply, SDOT divides the length of the legal parking zone along the block-face into 25 foot sections.

\begin{figure}[h]
    \begin{subfigure}[t]{0.45\textwidth}
        \centering
		\includegraphics[width=0.9\textwidth]{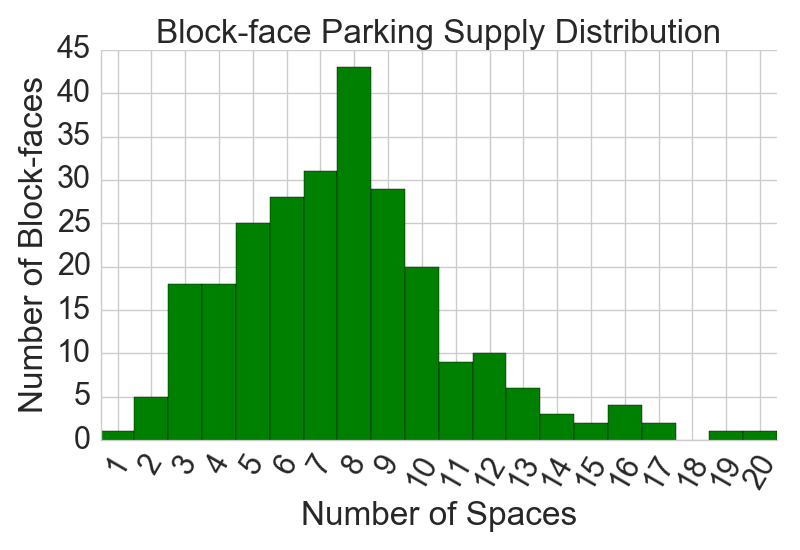}
        \caption{}
        \label{fig:supplydist}
    \end{subfigure}\hfill
    \begin{subfigure}[t]{0.45\textwidth}
        \centering
        \includegraphics[width=0.9\textwidth]{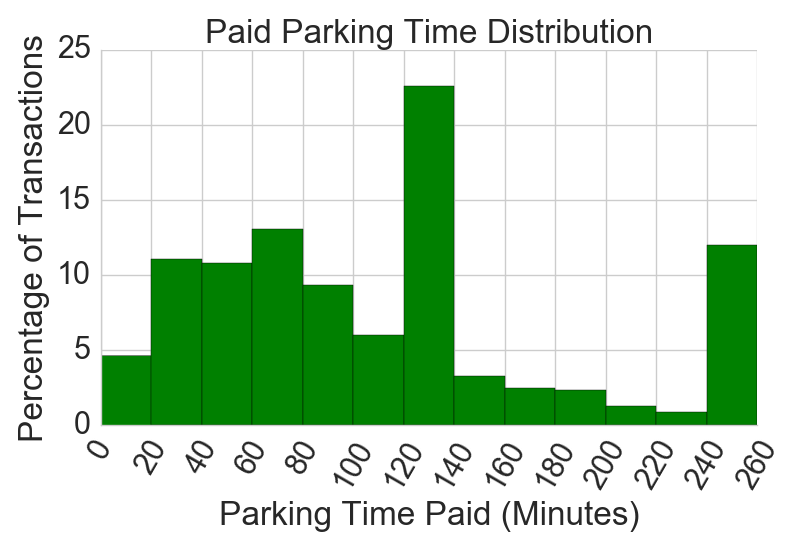}
        \caption{}
        \label{fig:servicedist}
	\end{subfigure}
	\caption{(a) Distribution of parking spaces per block-face in Belltown. (b)
		Distribution of Paid Parking Time.}
	\label{fig:margcost}
\end{figure}

Paid parking is active from 8 \textsc{AM}--8 \textsc{PM}, Monday through Saturday. As an exception to this, there are a select number of block-faces along downtown arterials in which no paid parking is allowed during portions of the morning and evening commutes to allow for more roadway capacity and for buses to stop. The pricing model for each block-face includes four separate rate intervals: 8 \textsc{AM} - 11 \textsc{PM} weekday, 11 \textsc{AM}--8 \textsc{PM} weekday, 8 \textsc{AM}--11 \textsc{PM} Saturday, and 11 \textsc{AM}--8 \textsc{PM} Saturday. Prices range between \$1.50 - \$2.50 per hour. The time limit for paid parking is two or four hours depending on location. From our data we observe that drivers typically park for the maximum allotted time allowed whether the limit is two or four hours (see Fig.~\ref{fig:servicedist}).

We measure occupancy by counting the number of spaces paid for at each block-face at each minute. We then convert the number of paid spots at each minute to a \emph{load} which is defined to be the number of spaces paid for at a block-face divided by the supply of the block-face, as estimated by SDOT. We then aggregate the loads to an hourly resolution as hourly occupancy is a typical performance metric. These loads do not give the true occupancy due to several categories of vehicles which may park curbside for free (e.g., disabled placard holders, government vehicles, car-sharing services). Further, the load can be greater than 1; this is the result of 1) cars leaving before their paid time is expired, and 2) SDOT's estimated 25 feet of parking space per vehicle being too large for small compact cars and motorcycles.

\subsection{Model solution for rejection rates from occupancy data}

The procedure for determining rejection rates from occupancy data in Belltown is described at length in \cite{dowling2017optimizing} and our code and data are available in our GitHub repository. In short, using Eqn.~\eqref{eqn:y_poly_occup_mu} and occupancy measures at each block-face, we solve for the arrival rates that attain those occupancy levels. We solve for the probability that the block-face is full in the Markovian service case, and use this to estimate the rate at which drivers are rejected. We simulate to compare these model estimates to rejection rates achieved with both exponential and fixed service times parameterized by the mean paid parking time at each block-face.

Amongst the strongest assumptions we make is that the \emph{exogenous} arrival process to the network is Poisson. We can test the validity of this assumption by looking at the inter-transaction times at block-faces across the network. These arrivals would constitute a subset of the total arrivals to the block-face as the arrival rate of drivers who \emph{were} able to find parking. Fig.~\ref{fig:interarrivalrate} illustrates an example exponential fit of inter-transaction times at a block-face within the Belltown network.

\begin{figure}[h]
	\centering
	\includegraphics[width=0.75\columnwidth]{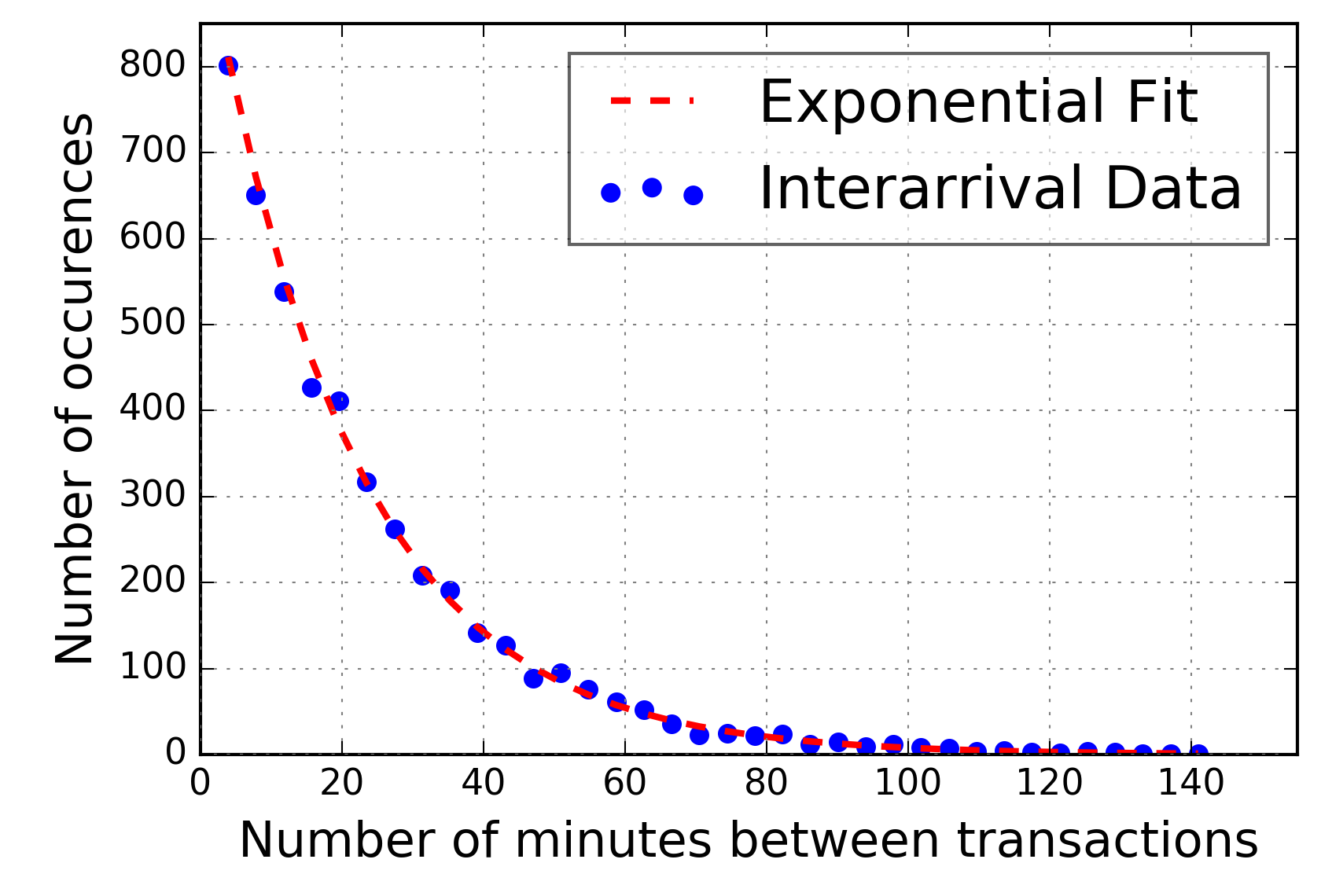}
	\caption{Exponential distribution of inter-transaction times at a block-face in Belltown over the course of January to June 2016.}
	\label{fig:interarrivalrate}
	\vspace{-10pt}
\end{figure}

A second key assumption in our application is that neighboring block-faces see similar levels of occupancy, and a growing body of work continues to justify this assumption that block-faces are spatially dependent upon one another \cite{fiez:2017aa}. In intuitive terms, it is unreasonable to expect that an 0\% occupancy block-face is immediately adjacent to a block-face that has larger than 90\% occupancy over some extended time period. 

We can evaluate this assumption by simulating such a network of finite capacity queues with arrival rate parameters learned from occupancy data. We can then corroborate the resulting 1) simulated occupancy level and 2) rejection rates between block-faces.

\subsection{Monte Carlo simulation of rate parameters}
\label{sec:simapp}

We simulate for each paid parking hour of each day the average occupancy observed in the time range of March 1$^{st}$ to July 30$^{th}$ to test the validity of our decoupling assumption in irregular networks. We measure for both fixed and exponential service times if 1) the exogenous arrival rates estimated by our model achieve the occupancies that we see in data and 2) the rejection rates of vehicles searching for parking is comparable to model estimates.

Simulation parameters are set according to the number of parking spaces, connectivity of the block-faces in Belltown, median paid parking time per block, and the hourly exogenous arrival rate estimated from occupancy by our above model. We average the results (per-block occupancies and rejection rates) of 100 simulations per day.

We treat each hour as if it were in steady state for much longer period of time, in order to obverse steady state values like the arrival rate, rejection rate, and occupancies. Therefore we allow the simulator to run over a much longer time horizon---1000 minutes---rather than the single hour that parameterizes its values.

\subsubsection{Exponential service}

\begin{figure}
    \centering
    \includegraphics[width=0.9\columnwidth]{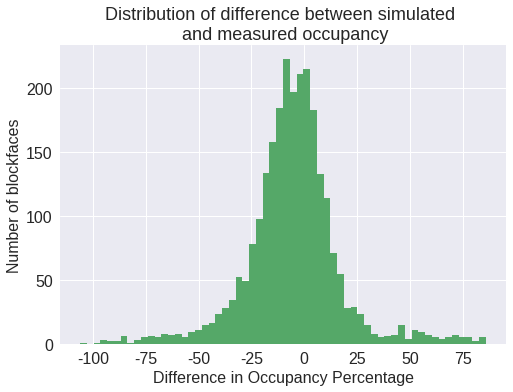}
    \caption{Distribution of differences between simulated occupancy and hourly occupancy from data for all days and hours}
    \label{fig:allsimerror}
\end{figure}

Here we expressly look at exponential service times in simulation to investigate the error incurred by solving for the rates that achieve the observed occupancy at each block-face individually. Fig.~\ref{fig:allsimerror} is a histogram of the differences between the simulated occupancy under exponential service and the observed occupancy in data. We achieve an average error of -5.3\% and a standard deviation of $\pm$22.3\%.

We observe a negative mean likely due to our conservative method of estimation of the total arrival rates to a block-face. The queue network model yeilds an asymptotic relationship between occupancy and arrival rate $y$ achieving that occupancy.  Because of the monotonocity of this relationship \cite{dowling2017optimizing}, when finding the root of Eqn.~\eqref{eqn:y_poly_occup_mu} our implementation performs a brute force search in increasing $y$ subject to a tolerance parameter on the occupancy $u$. When estimating $y$ as a function of $u$ we hard thresholded occupancy $u$ at 99\%. 

\begin{figure}
    \centering
    \includegraphics[width=0.9\columnwidth]{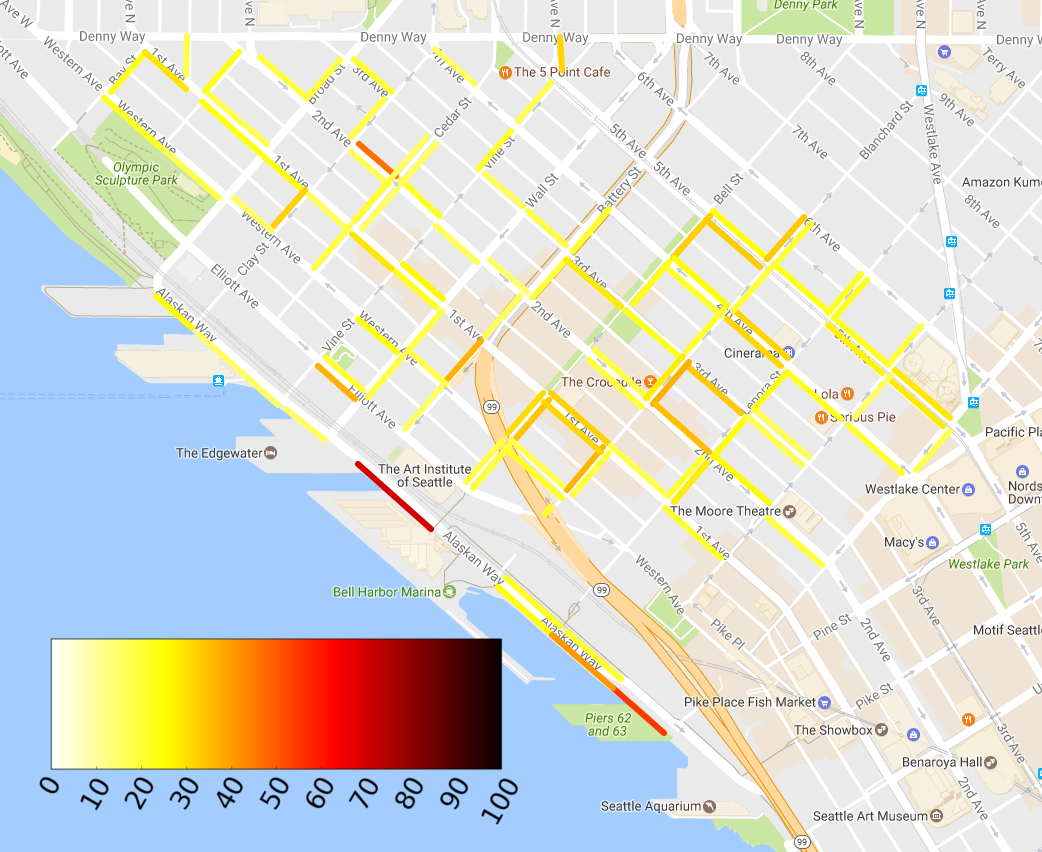}
    \caption{Locations of average occupancy error between data and simulation across weekdays, March --- June 2016}
    \label{fig:simerrorloc}
\end{figure}

Block-faces exhibiting the most error were consistent across time and space, suggesting the model is efficient for normal cases. Indeed, the top three outliers---block-faces 91, 185, and 196\footnote{See \url{github.com/cpatdowling/net-queue/data/simulation/belltownsims/belltowndata/data_notes.txt} for additional information}---showed consistent discrepancy between simulated and observed occupancy due to, we believe, the following reasons: 1) block-face 91 (farthest east red block-face along the waterfront in Fig.~\ref{fig:simerrorloc}) had the minimum mean paid parking time of Q2 2016 in Belltown (48 minutes vs 109 minutes overall average, noting the strong tendency toward paying for the maximum parking time as observed in Fig.~\ref{fig:servicedist}, farthest west red block-face along the waterfront in Fig.~\ref{fig:simerrorloc}), 2) block-face 185 is on the boundary of an isolated section of the network, connected to only one other block-face, and 3) block-face 196 is the only block-face in Belltown with 1 space of paid parking (see Fig.~\ref{fig:supplydist}, farthest north red block in Fig.~\ref{fig:simerrorloc}). 

\begin{figure}
    \centering
    \includegraphics[width=0.9\columnwidth]{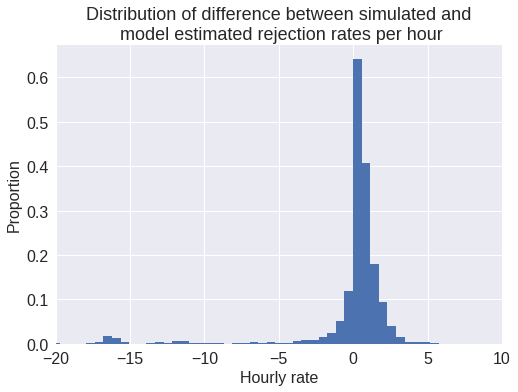}
    \caption{Distribution of differences between simulated rejection rates and model estimated rejection rates from data for all days and hours}
    \label{fig:ratehist}
\end{figure}

Over all days, we observed an average of error of -0.19 vehicles per hour in rejection rates (between model estimate and simulated) with a standard deviation of $\pm$4 vehicles per hour across block-faces for which rejections were detectable at rates less than the simulation time horizon (1000 minutes). This would suggest that, according to simulation or average model estimate of rejected vehicles circulating in Belltown looking for parking is conservative. Fig.~\ref{fig:ratehist}, however, illustrates the differences at all days and hours; the opposite may be the case, having a median of 0.46 vehicles per hour. Block-faces 0, 76, and 90 were consistently underestimated by a wide margin (less than -15 vehicles per hour in some cases), but for potential reasons less clear than error observed in simulated occupancy. The maximum \emph{over-estimation} was roughly 5.7 vehicles per hour.

\subsubsection{Fixed Service}

For the fixed service case, we achieve similar results. Here we make a direct comparison with exponential service on a typical Monday during March --- June 2016, measuring again the simulated occupancy error from the true measured occupancy, illustrated in Fig.~\ref{fig:mondayerror}. Across all days and times, we achieves a slightly lower, -5.3\% average error in occupancy and a standard deviation of $\pm$21.2\%, with an identical set of the top three outlier block-faces in terms of error.

\begin{figure}
    \centering
    \includegraphics[width=0.9\columnwidth]{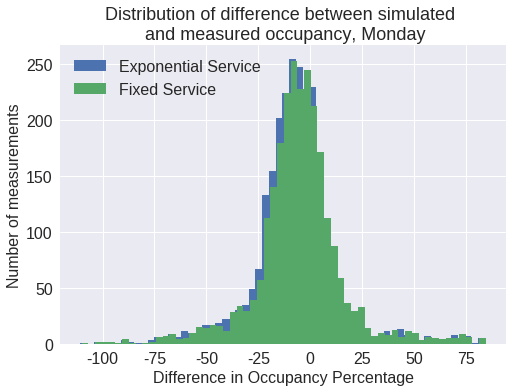}
    \caption{Distribution of differences between simulated occupancy and true occupancy data for Monday, all hours, for both fixed and exponential service times}
    \label{fig:mondayerror}
\end{figure}

Critical to the utility of this network queue model is the ability to estimate the rate at which vehicles are rejected by a full block-face, as these are the vehicles contributing to through-traffic cruising for an available space. Again, restricting our comparison between exponential and fixed service times to Monday, we again achieve similar results, illustrated in Fig.~\ref{fig:rateerrorMonday}. The fixed service time simulation exhibits a slightly higher mean rejection rate error on Monday than exponential service (0.38 vehicles per hour vs 0.34 vehicles per hour); and again, outliers (maximum rate errors of 3.26 (fixed service) and 3.28 (exponential service)) were due to the same boundary case block-faces.

\begin{figure}
    \centering
    \includegraphics[width=0.9\columnwidth]{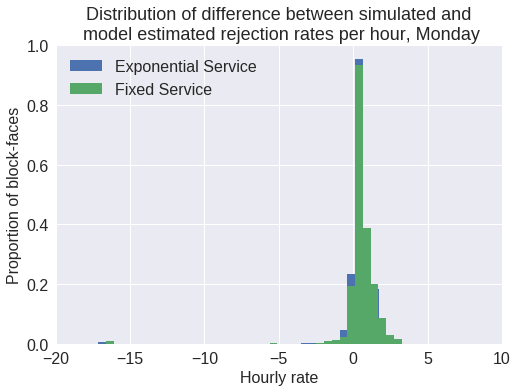}
    \caption{Distribution of differences between simulated rejection rates and model estimated rejection rates from data for Monday, all hours, for both fixed and exponential simulated service times}
    \label{fig:rateerrorMonday}
\end{figure}

%% file: conclusion.tex
In sum, we have investigated a number of assumptions necessary to model curbside parking as a network of finite capacity queues, demonstrating a new model for urban parking that considers spatial and temporal heterogeneity. This model provides a means for city planners and traffic engineers to analyze a network of curbside parking at a much higher resolution that previously possible. These results help provide a basis for future work designing and exploring new parking scheduling and pricing regimes to minimize delays caused by cruising for parking while leaving open a flexible resource to drivers and businesses.